\documentclass[11pt,a4]{article}

\usepackage[USenglish]{babel}
\usepackage[utf8]{inputenc}
\usepackage[T1]{fontenc}
\usepackage[inline]{enumitem}
\usepackage{xspace}
\usepackage{amsmath}
\usepackage{amsthm}
\usepackage{amssymb}
\usepackage{xcolor}
\usepackage[colorlinks=True,
            citecolor=blue,
            urlcolor=violet,
            linkcolor=purple]{hyperref}
\usepackage{doi}
\usepackage[nameinlink, capitalize]{cleveref}
\usepackage[style=alphabetic, maxnames=99]{biblatex}
\usepackage{tikz}
\usepackage{pgfplots}
\usepackage{tcolorbox}
\usepackage{manyfoot}
\usepackage{algorithm}
\usepackage{algorithmicx}
\usepackage{algpseudocode}
\usepackage[textwidth=140mm,
			textheight=222mm,
			top=25truemm,
			bottom=25truemm,
			left=30truemm,
			right=30truemm]{geometry}
\usepackage{lineno}

\usetikzlibrary{plotmarks}
\usepgfplotslibrary{fillbetween}
\pgfplotsset{compat=1.16}

\tcbuselibrary{breakable}
\tcbuselibrary{skins}

\crefname{lemma}{Lemma}{Lemmata}
\crefname{figure}{Figure}{Figures}

\newcommand{\email}[1]{\href{mailto:#1}{\nolinkurl{#1}}}

\providecommand{\keywords}[1]
{
  \small
  \begin{description}
  	\item [Keywords:] #1
  \end{description}
}

\newtheorem{theorem}{Theorem}
\theoremstyle{plain}
\newtheorem{lemma}[theorem]{Lemma}
\newtheorem{corollary}[theorem]{Corollary}
\newtheorem{proposition}[theorem]{Proposition}
\theoremstyle{definition}
\newtheorem{definition}[theorem]{Definition}
\newtheorem{remark}[theorem]{Remark}

\def\poly{{\mathrm{poly}}}

\def\cc#1{\mathtt{#1}}
\DeclareMathOperator{\FNP}{\cc{FNP}}

\DeclareMathOperator{\NP}{\cc{NP}}
\DeclareMathOperator{\TFNP}{\cc{TFNP}}
\DeclareMathOperator{\TFAP}{\cc{TFAP}}
\DeclareMathOperator{\PPP}{\cc{PPP}}
\DeclareMathOperator{\PWPP}{\cc{PWPP}}
\DeclareMathOperator{\PLC}{\cc{PLC}}
\DeclareMathOperator{\PLS}{\cc{PLS}}
\DeclareMathOperator{\PPADS}{\cc{PPADS}}

\def\problem#1{\textsc{#1}\xspace}

\def\Pigeon{\problem{Pigeon}}
\def\WeakPigeon{\problem{Weak Pigeon}}
\def\WardSzabo{\problem{Ward-Szab{\'o}}}
\def\BasicWardSzabo{\problem{Basic} \WardSzabo}

\def\CategorizedPigeon{\problem{Categorized Pigeon}}
\def\AlternativeCategorizedPigeon{\problem{Alternative} \CategorizedPigeon}
\def\WeakAlternativeCategorizedPigeon{\problem{Weak Alternative Categorized Pigeon}}
\def\InjectiveCategorizedPigeon{\problem{Injective} \CategorizedPigeon}
\def\SinkofLine{\problem{Sink of Line}}
\def\MultiSourceSinkofLine{\problem{Multiple Source} \SinkofLine}
\title{
    The $\PPP$-completeness of the Ward-Szab\'{o} theorem
}

\author{
	Takashi Ishizuka\\
	National Institute of Technology, Kochi Colledge\\
	\email{ishizuka@kochi-ct.ac.jp}
}

\begin{document}

\begin{titlepage}
    \clearpage
    
    \maketitle
    
    \begin{abstract}
        \citeauthor{WS94} \cite{WS94} have shown that a complete graph with $N^2$ nodes whose edges are colored by $N$ colors and that has at least two colors contains a bichromatic triangle.
        This fact leads us to a total search problem: Given an edge-coloring on a complete graph with $N^2$ nodes using at least two colors and at most $N$ colors, find a bichromatic triangle.
        \citeauthor{BFHRS23} \cite{BFHRS23} have proven that such a total search problem, called \WardSzabo, is $\PWPP$-hard and belongs to the class $\TFNP$, a class for total search problems in which the correctness of every candidate solution is efficiently verifiable.
        However, it is open which $\TFNP$ subclass contains \WardSzabo.
        This paper will improve the computational complexity of \WardSzabo.
        We prove that \WardSzabo is a complete problem for the complexity class $\PPP$, a $\TFNP$ subclass of problems in which the existence of solutions is guaranteed by the pigeonhole principle.
    \end{abstract}
    
    \keywords{$\TFNP$, $\PPP$, the pigeonhole principle, the Ward-Szab{\'o} theorem}
    
    \thispagestyle{empty}
\end{titlepage}


\section{Introduction}
\label{Sec:Intro}
A \textit{total search problem} is a problem for finding a solution whose existence is guaranteed by some mathematical lemma.
The complexity class $\TFNP$ consists of all total search problems for which the correctness of every potential solution is effortlessly checkable.
It is well known that $\TFNP$ contains many important search problems that we want to solve efficiently, but at present, no one knows how to do so; for example, computing a mixed Nash equilibrium for a two-player game and computing a prime divisor for a positive integer.

We are interested in the computational aspects of extremal combinatorics on a graph with exponentially many nodes.
This paper focuses on the computational aspects of Ward-Zab\'{o} theorem \cite{WS94}.
Let $K_N$ be the complete graph with $N$ nodes.
An edge-coloring function $C: E(K_{N}) \to [m]$ on $K_N$ is said to be \textit{swell coloring} if $C$ uses at least two colors and there is no birchomatic triangle; each triangle contains exactly one or three colors.
\citeauthor{WS94} \cite{WS94} have proven that we need at least $\sqrt{N} + 1$ colors to be swell coloring on $K_{N}$.

\begin{theorem}[\citeauthor{WS94} \cite{WS94}] \label{Theorem:Ward-Szabo}
    The complete graph on $N$ nodes cannot be swell-colored with fewer than $\sqrt{N} + 1$ colors, and this bound is tight.
\end{theorem}

\cref{Theorem:Ward-Szabo} leads us to a total search problem for finding a bichromatic triangle on a complete graph when we are given an edge-coloring with an insufficiency number of colors.
Recently, \citeauthor{BFHRS23} \cite{BFHRS23} has formulated such a $\TFNP$ problem, called \WardSzabo.

\begin{definition} \label{Def:WardSzaboProblem}
    The problem \WardSzabo, is defined as follows: Given a Boolean circuit $C: \{ 0, 1 \}^{2n} \times \{ 0, 1 \}^{2n} \to \{ 0, 1 \}^n$ that computes an edge-coloring for every edge on the complete graph $K_{2^{2n}}$, and three distinct nodes $a, b, c \in \{ 0, 1 \}^{2n}$ with $C(a, b) \neq C(a, c)$, find one of the following:
    \begin{enumerate*}[label = (\arabic*)]
        \item three distinct nodes $x, y, z \in \{ 0, 1 \}^{2n}$ such that $C(x, y) = C(x, z) \neq C(y, z)$;
        \item two distinct nodes $x, y \in \{ 0, 1 \}^{2n}$ such that $C(x, y) \neq C(y, x)$, i.e., a violation that $C$ is symmetric.
    \end{enumerate*}
\end{definition}

\noindent
For every \WardSzabo instance, we consider an edge coloring using at least two colors and at most $2^{n}$ colors on the complete graph with $2^{2n}$ nodes.
From \cref{Theorem:Ward-Szabo}, we know that such a coloring cannot be a swell coloring . Hence, there must exist a triangle that contains just two colors.

\citeauthor{BFHRS23} \cite{BFHRS23} have proven that \WardSzabo lies on between the classes $\PWPP$ and $\TFNP$.
This paper improves the computational complexity of this problem.
Specifically, we show that \WardSzabo is in the class $\PPP$, a $\TFNP$ subclass of the problem whose totality is guaranteed by the \textit{pigeonhole principle}.

\begin{theorem} \label{main_theorem}
    The problem \WardSzabo is $\PPP$-complete.
\end{theorem}

\subsection{Related Work}
\citeauthor{BFHRS23} \cite{BFHRS23} have investigated several $\TFNP$ problems based on extremal combinatorics.
\cite{BFHRS23} has left some interesting open questions worth considering. The complexity of \WardSzabo is one of them.
\citeauthor{BGS25} \cite{BGS25} have studied the complexity of problems related to the coding theory and extremal combinatorics.
\citeauthor{Ghentiyala25} \cite{Ghentiyala25} has studied the connection between the $\TFNP$ problems from extremal combinatorics and the recently discovered $\TFNP$ subclass $\TFAP$ introduced by \citeauthor{Li24} \cite{Li24}.
\cite{Ghentiyala25} has also shown the $\PLC$ (Polynomial Long Choice) containment of \WardSzabo and its $\PPP$-hardness.

It is known that the class $\PPP$ has a relationship with cryptography.
For instance, \citeauthor{Pap94} \cite{Pap94} has shown that the existence of a polynomial-time algorithm for a $\PPP$-complete problem implies that there is no one-way permutation.
\citeauthor{SZZ18} \cite{SZZ18} and \citeauthor{HV21} \cite{HV21} have studied the relationship the classes $\PPP$ and $\PWPP$ with total search problems inspired by cryptography.
The $\PPP$-completeness of a constrained version of \problem{Short Integer Solution} has been shown in \cite{SZZ18}. The $\PWPP$-completeness of \problem{Discrete Logarithm Problem} has been shown in \cite{HV21}. 

Recently, the computational complexity of some generalizations of the class $\PPP$ has been studied \cite{JLRX24, Ishizuka25, PPY23}.
\citeauthor{PPY23} \cite{PPY23} have introduced a generalization of $\PPP$ in which it is iteratively applied by the pigeonhole principle; they call such a complexity class $\PLC$.
\citeauthor{JLRX24}, \cite{JLRX24} have introduced a generalization of $\PPP$ in which we find a multiple collision, and they call the principle that guarantees the existence of a multiple collision the \textit{peaking order principle}. Moreover, \cite{JLRX24} has introduced the pigeonhole hierarchy.
\citeauthor{Ishizuka25} \cite{Ishizuka25} has considered the complexity of the pigeonhole principle over a quotient set and introduced a new $\TFNP$ subclass that contains not only $\PPP$, but also $\PLS$, which is a class of search problems that can be solved by a local search method.

Finally, we mention a key property of the complexity class $\PPP$.
\citeauthor{FGPR24} \cite{FGPR24} has proven that $\PPP$ is unlikely closed under the Turing reduction.

\section{Preliminaries}
\label{Sec:Preliminaries}
We denote by $\mathbb{Z}$ the set of all integers.
For every integer $a \in \mathbb{Z}$ and for each binary relation $\square \in \{ <, \le, >, \ge \}$, we define $\mathbb{Z}_{\square a} := \{ x \in \mathbb{Z} : x \square a \}$.

Let $A, B$ be finite sets, and let $f: A \to B$ be a function.
We denote by $|A|$ the cardinality of the set $A$, i.e., the number of elements in $A$.
For each subset $C \subseteq A$, we denote the set $\{ f(a) : \forall a \in A \} \subseteq B$ as $f(C)$.

Let $\{ 0, 1 \}^*$ denote the set of binary strings with a finite length.
For each string $x \in \{0, 1\}^*$, we denote by $|x|$ the length of $x$.
For each positive integer $n \in \mathbb{Z}_{> 0}$, we write $\{ 0, 1 \}^n$ for the set of binary strings with the length $n$.
Throughout this paper, we sometimes partition a string into two or more parts of strings. 
For example, we sometimes split an $(m_1 + m_2)$-bit string $x$ into a pair of an $m_1$-bit string $y_1$ and an $m_2$-bit string $y_2$ for all two positive integers $m_1$ and $m_2$ in $\mathbb{Z}_{> 0}$, and we denoted as $x = y_1 \| y_2$ or $x = (y_1, y_2)$, where $y_1$ and $y_2$ are the first $m_1$ and last $m_2$ bits of $x$, respectively.

\subsection{Total Search Problems in $\NP$}
\label{Sec:TFNP}
Let $R \subseteq \{ 0, 1 \}^* \times \{ 0, 1 \}^*$ be a relation.
We say that $R$ is \textit{polynomially balanced} if there is a polynomial $p: \mathbb{Z}_{\ge 0} \to \mathbb{Z}_{\ge 0}$ such that for each $(x, y) \in R$, it holds that $|y| \le p(|x|)$.
We say that $R$ is \textit{polynomial-time decidable} if for each pair of strings $(x, y) \in \{ 0, 1 \}^* \times \{ 0, 1 \}^*$, we can decide whether $(x, y) \in R$ in polynomial time.
We say that $R$ is \textit{total} if for each string $x \in \{0, 1\}^*$, there is at least one string $y \in \{ 0, 1 \}^*$ such that $(x, y) \in R$.

For a relation $R \in \{ 0, 1 \}^* \times \{ 0, 1 \}^*$, the search problem with respect to $R$ is defined as follows: Given a string $x \in \{ 0, 1 \}^*$, find a string $y \in \{ 0, 1 \}^*$ such that $(x, y) \in R$ if such a $y$ exists, otherwise report ``No.'' For simplicity, we call the search problem with respect to $R$ the search problem $R$.
If $R$ is \textit{total}, we call such a search problem a total search problem.
The complexity class $\FNP$ is the set of all search problems with respect to polynomially balanced and polynomial-time decidable relations. The complexity class $\TFNP$ is the set of all total search problems belonging to $\FNP$. By definition, it holds that $\TFNP \subseteq \FNP$.

Let $R, S \subseteq \{ 0, 1 \}^* \times \{ 0, 1 \}^*$ be two total search problems. A polynomial-time reduction from $R$ to $S$ is defined by two polynomial-time computable functions $f: \{ 0, 1 \}^* \to \{ 0, 1 \}^*$ satisfying that $(x, g(x, y)) \in R$ whenever $(f(x), y) \in S$. Thus, the function $f$ maps an input $x$ of $R$ to the input $f(x)$ of $S$, and the other function $g$ maps a solution $y$ to the input $f(x)$ for the search problem $S$ to a solution $g(x, y)$ to the input $x$ for the search problem $R$.
We say that two search problems $R$ and $S$ are (polynomially) equivalent if there are polynomial-time reductions from $R$ to $S$ and from $S$ to $R$.

Let $\mathcal{C}$ be a complexity class, and let $R \subseteq \{ 0, 1 \}^* \times \{ 0, 1 \}^*$ be a search problem.
We say that $R$ is $\mathcal{C}$-hard if there is a polynomial-time reduction from $S$ to $R$ for every search problem $S$ in $\mathcal{C}$.
The problem $R$ is said to be $\mathcal{C}$-complete if $R$ is $\mathcal{C}$-hard and belongs to the class $\mathcal{C}$.

\paragraph{Complexity Classes $\PPP$ and $\PWPP$}
\label{Sec:PPP}
\citeauthor{Pap94} \cite{Pap94} has introduced a $\TFNP$ subclass based on the pigeonhole principle, called the complexity class $\PPP$.
This class is defined as the set of all search problems that are reducible to the search problem \Pigeon (see \cref{Def:Pigeon}) in polynomial time.

The pigeonhole principle (also known as \textit{Dirichlet's principle}) \cite{Jukna11} states that for all positive integers $n, m, \ell$ with $n > m \ell$, when we partition the set $\{ 1, 2, \dots, n \}$ into $m$ groups, there always exists a group with at least $\ell+1$ elements.

The totality of the canonical $\PPP$-complete problem \Pigeon is based on the pigeonhole principle, where $\ell = 1$.
For every function $f$ that maps each $n$-bit string to an $n$-bit string, we have a string $x$ such that $f(x) = v^*$ if the function $f$ is bijective; otherwise, we have a collision, i.e. two distinct strings $x$ and $y$ such that $f(x) = f(y)$ by the pigeonhole principle.

\begin{definition} \label{Def:Pigeon}
    The problem \Pigeon is defined as follows: Given a Boolean circuit $f: \{ 0, 1 \}^n \to \{ 0, 1 \}^n$ and a special element $v^* \in \{ 0, 1 \}^*$, find one the following:
    \begin{enumerate*}[label = (\arabic*)]
        \item an element $x \in \{ 0, 1 \}^*$ such that $f(x) = v^*$
        \item two distinct elements $x, y \in \{ 0, 1 \}^*$ such taht $f(x) = f(y)$.
    \end{enumerate*}
\end{definition}

\citeauthor{Jer16} \cite{Jer16} has introduced the complexity class $\PWPP$, which is a weaker variant of the complexity class $\PPP$. 
This class is defined as the set of all search problems that are reducible to the search problem \WeakPigeon (see \cref{Def:WeakPigeon}) in polynomial time.
By definition, it is not hard to see that $\PWPP$ is a subclass of $\PPP$.

\begin{definition} \label{Def:WeakPigeon}
    The problem \WeakPigeon is defined as follows: Given a Boolean circuit $f: \{ 0, 1 \}^{n} \to \{ 0, 1 \}^{n-1}$, find two distinct elements $x$ and $y$ in $\{ 0, 1 \}^{n}$ such that $f(x) = f(y)$.
\end{definition}

Finally, \citeauthor{Jer16} \cite{Jer16} has also introduced a useful parameterized formulation of \Pigeon.
For any $A(n) > B(n)$, the problem $\Pigeon_{B(n)}^{A(n)}$ is defined as follows: Given a Boolean circuit $f:[A(n)] \to [B(n)]$, which has $\lceil \log_2A(n) \rceil$ inputs and $\lceil \log_2 B(n) \rceil$ outputs, find two distinct elements $x$ and $y$ in $[A(n)]$ such that $f(x) = f(y)$.
From the formulation, \WeakPigeon is equivalent to $\Pigeon_{2^{n-1}}^{2^{n}}$.
More precisely, the following statement holds.

\begin{proposition}[\citeauthor{Jer16} \cite{Jer16}] \label{Prop:pigeon_equiv}
    The next two relationships holds:
    \begin{enumerate}[label = (\roman*)]
        \item \Pigeon is equivalent to $\Pigeon_{2^n}^{2^{n}+1}$; and
        \item \WeakPigeon is equivalent to $\Pigeon_{B(n)}^{A(n)}$ whenever $A(n) \ge (1 + \poly(\log B(n)))B(n)$.
    \end{enumerate}
\end{proposition}

\section{Our Results}
\label{Sec:Results}
To prove \cref{main_theorem}, we show the following two lemmata:

\begin{lemma} \label{Lemma:main_containment}
    \WardSzabo belongs to $\PPP$
\end{lemma}

\begin{lemma} \label{Lemma:main_hardness}
    \WardSzabo is $\PPP$-hard.
\end{lemma}

\noindent
\cref{main_theorem} immediately follows from these two above lemmata.
The proofs of \cref{Lemma:main_containment,Lemma:main_hardness} can be found in \cref{Sec:UpperBound,Sec:Hardness}, respectively.

\subsection{A Better Upper Bound of \WardSzabo}
\label{Sec:UpperBound}
For proving \cref{Lemma:main_containment}, we first observe an usefule property of \WardSzabo, and we introduce a slightly harder $\TFNP$ problem than the original \WardSzabo formulated in \cite{BFHRS23}.
For every \WardSzabo instance $\langle C: \{ 0, 1 \}^{2n} \to \{ 0, 1 \}^{n}; \{ a, b, c \} \rangle$, we say that a triangle $\{ x, y, z \}$ on the complete graph $K_{2^{2n}}$ is \textit{basic} if $\{ x, y, z \} \cap \{ a, b, c \} \neq \emptyset$.

By the pigeonhole principle, we show that every \WardSzabo instance has a basic bichromatic triangle.
Intuitively, this follows from the following properties.
Since the node $a$ is incident with $2^{2n} - 1$ edges and these edges are colored with $2^{n}$, there are at least $2^n$ edges of $a$ that are colored by the same color.
Note that we also know that the colors of two edges $\{ a, b \}$ and $\{ a, c \}$ are different.
Again, by the pigeonhole principle, we can guarantee the existence of a basic bichromatic triangle.

\begin{proposition} \label{Prop:BasicProperty}
	Let $K_{N^2}$ be a complete graph with $N^2$ nodes, and let $a, b, c$ be distinct three nodes on $K_{N^2}$.
	For every edge-coloring function $C: E(K_{N^2}) \to [N]$ such that $C(a, b) \neq C(a, c)$, there is a bichromatic triangle  $\{ x, y, z \}$ such that $\{ a, b, c \} \cap \{ x, y, z \} \neq \emptyset$.
\end{proposition}
\begin{proof}
	Snce $K_{N^2}$ is the complete graph with $N^2$ nodes, the node $a$ is connected with $N^2 - 1$ nodes.
	By the pigeonhole principle, we can see that there is a color $\chi \in [N]$ such that at least $N$ edges that are incident on $a$ have the color $\chi$. 
	Let $V_{\chi} := \{ v \in V(K_{N^2}) \setminus \{ a \} : C(\{ a, v \}) = \chi \}$ be the set of nodes connected with $a$ by the color $\chi$.
	If there are two nodes $y, z \in V_{\chi}$ such that $C(\{ y, z \}) \neq \chi$, then we obtain a basic bichromatic triangle $\{ a, y, z \}$ satasfying that $C(\{ a, y \}) = C(\{ a, z \}) \neq C(\{ y, z \})$.
	
	Below, we suppose that each pair of two nodes $y, z$ in $V_{\chi}$ satisfies that $C(\{ y, z \}) = \chi$.
	From the assumption that $C(\{ a, b \}) \ne C(\{ a, c \})$, at least one of $b$ and $c$ is outside of $V_{\chi}$. Without loss of generality, we can assume that $b \notin V_{\chi}$.
	Again, by the pigeonhole principle, we have two distinct nodes $y, z$ in $V_{\chi} \cup \{ a, \}$ such that $C(\{ b, y \}) = C(\{ b, z \})$. At least one of $y$ and $z$ is not $a$. Without loss of generality, $y \neq a$.
	If $C(\{ b, y \}) = \chi$, then the triangle $\{ a, b, y \}$ is a bichromatic triangle since $C(\{ a, b \}) \neq C(\{ a, y \})$.
	Otherewise, it holds that $C(\{ b, y \}) = C(\{ b, z \}) \neq C(\{ y, z \})$. This implies that the triangle $\{ b, y, z \}$ is bichromatic.
\end{proof}

\cref{Prop:BasicProperty} implies that there always exists a basic triangle that is bichromatic for every \WardSzabo instance.
This leads us to formulate a slightly harder $\TFNP$ problem related to Ward-Szab\'{o} theorem, the problem of finding a basic bichromatic triangle on the complete graph with $2^{2n}$ nodes. We call such a problem \BasicWardSzabo.

\begin{definition} \label{Def:BasicWardSzaboProblem}
    The problem \BasicWardSzabo is defined as follows: Given a Boolean circuit $C: \{ 0, 1 \}^{2n} \times \{ 0, 1 \}^{2n} \to \{ 0, 1 \}^n$, and three distinct nodes $a, b, c \in \{ 0, 1 \}^{2n}$ with $C(a, b) \neq C(a, c)$, find one of the following:
    \begin{enumerate*}[label = (\arabic*)]
        \item three distinct nodes $x, y, z \in \{ 0, 1 \}^{2n}$ such that $\{ x, y, z \} \cap \{ a, b, c \} \neq \emptyset$ and $C(x, y) = C(x, z) \neq C(y, z)$;
        \item two distinct nodes $x, y \in \{ 0, 1 \}^{2n}$ such that $C(x, y) \neq C(y, x)$, i.e., a violation that $C$ is symmetric.
    \end{enumerate*}
\end{definition}

\begin{lemma} \label{Lemma:WarSzabo-to-BasicWardszabo}
	\WardSzabo is reducible to \BasicWardSzabo in polynomial time.
\end{lemma}
\begin{proof}
	Every solution to \BasicWardSzabo is also a solution to \WardSzabo.
\end{proof}

Therefore, it suffices to show the $\PPP$ containment of \BasicWardSzabo in order to prove \cref{Lemma:main_containment}.
We now introduce a new variant of \Pigeon, called \CategorizedPigeon, that helps us to construct a reduction from \BasicWardSzabo to a $\PPP$ problem.
Our new variant of \Pigeon is inspired by two previous works \cite{JLRX24, PPY23} in the following sense:

\begin{description}
	\item [the peaking order principle] \citeauthor{JLRX24} \cite{JLRX24} have introduced a new $\TFNP$ problem based on the pigeonhole principle when $\ell$ is greater than one. Our new variant also depends on the pigeonhole principle with $\ell > 1$.
	\item [the long choice principle] \citeauthor{PPY23} \cite{PPY23} have introduced a new $\TFNP$ problem that is iteratively applied by the pigeonhole principle. Our new variant applies the pigeonhole principle twice.
\end{description}

\begin{definition}
    The problem $\CategorizedPigeon$ is defined as follows: Given two Boolean circuits $p: \{ 0, 1 \}^{m_1 + m_2} \to \{ 0, 1 \}^{m_1}$ and $h: \{ 0, 1 \}^{m_1 + m_2} \to \{ 0, 1 \}^{m_2}$ and a list of $k$ removed elements $\Pi := \{ \pi_1, \pi_2, \dots, \pi_k \} \subseteq \{ 0, 1 \}^{m_1 + m_2}$, where $k < 2^{m_1}$, find one of the following:
    \begin{enumerate*}[label = (\arabic*)]
        \item an element $x \in \{ 0, 1 \}^{m_1 + m_2} \setminus \Pi$ such that $h(x) \in h(\Pi)$;
        \item two distinct elements $x, y \in \{ 0, 1 \}^{m_1 + m_2} \setminus \Pi$ such that $p(x) = p(y)$ and $h(x) = h(y)$.
    \end{enumerate*}
\end{definition}

The totality of the problem \CategorizedPigeon is based on the pigeonhole principle.
If the set with $2^{m_1 + m_2} - k$ elements is partitioned into $2^{m_1}$ groups, where $k < 2^{m_1}$, then there is a group with at least $2^{m_2}$ elements.
For every \CategorizedPigeon instance, the function $p$ implicitly represents a partition the set $\{ 0, 1 \}^{m_1 + m_2} \setminus \Pi$ into $2^{m_2}$ groups $M_1, M_2, \dots, M_{2^{m_2}}$. Hence, $p(x)$ implies that the element $x$ is partitioned into the $p(x)$-th group.
Again, by the pigeonhole principle, every map $h$ from the set with $2^{m_2}$ elements to another set with $2^{m_2} - 1$ elements always has a collision, i.e., two distinct elements $x$ and $y$ such that $h(x) = h(y)$.
Roughly speaking, for every \CategorizedPigeon instance, the function $h$ represents a \Pigeon instance with $m_2$ input and output gates, respectively.

\begin{remark}
    By the definition of \CategorizedPigeon, we are explicitly given a list of at most $2^{m_1} - 1$ removed elements.
    Thus, we acknowledge that the input length becomes exponentially large, but this is not an issue.
    When an instance has exponentially many removed elements, we can easily solve such an instance since the input length is sufficiently large.
    Therefore, in practice, we are only interested in instances that have at most polynomially many removed elements.
    Note that in the case where $m_1$ is relatively very small than $m_2$ such as $m_1 \le \log m_2$ or $m_1$ is a constant, the input length of the problem is polynomial in $m_2$ even if we are given $2^{m_1} - 1$ removed elements.
\end{remark}

We observe some basic computational complexity aspects of \CategorizedPigeon in \cref{Sec:BasicProperties_of_CPPP}.
In this section, we present some key properties of \CategorizedPigeon that we need to prove \cref{Lemma:main_containment}.
Unfortunately, it is unclear whether \CategorizedPigeon is a $\PPP$-complete problem. More precisely, we can straightforwardly see the $\PPP$-hardness of \CategorizedPigeon, but we are unaware of the $\PPP$ containment of this $\TFNP$ problem.
Nevertheless, it is $\PPP$-complete if an instance has exactly one removed element, i.e., $k = 1$.
The next lemma states that we have a polynomial-time reduction from \CategorizedPigeon with $k = 1$ to the $\PPP$-complete problem \Pigeon.

\begin{lemma} \label{Lemma:CategorizedOne-in-PPP}
    \CategorizedPigeon with exactly one removed element is in the class $\PPP$.
\end{lemma}
\begin{proof}
    Let $\langle p: \{ 0, 1 \}^{m_1 + m_2} \to \{ 0, 1 \}^{m_1}; h: \{ 0, 1 \}^{m_1 + m_2} \to \{ 0, 1 \}^{m_2}; \Pi := \{ \pi \} \rangle$ be an instance of \CategorizedPigeon with exactly one removed element.
    We now construct a \Pigeon instance.
    
    Let $n = m_1 + m_2$.
    First, we define the special element $v^*$ as $0^{m_1} \| h(\pi)$ in $\{ 0, 1 \}^{n}$.
    Next, we define the Boolean function $f: \{ 0, 1 \}^{n} \to \{ 0, 1 \}^{n}$ as follows: For each $x \in \{ 0, 1 \}^{n}$,
    
    \begin{align*}
        f(x) := \begin{cases}
            1^{m_1} \| h(\pi) & \text{ if } x = \pi \\
            p(x) \| h(x) & \text{ otherwise.}
        \end{cases}
    \end{align*}

    \noindent
    Thus, for every $x \neq \pi$, the first $m_1$ bits and the last $m_2$ bits of $f(x)$ are equal to $p(x)$ and $h(x)$, respectively.
    
    We have completed constructing the reduced instance of \Pigeon. It is not hard to see that such a reduction can be computed in polynomial time.
    What remains is to prove that we can effortlessly obtain a solution to the original instance of \CategorizedPigeon with exactly one removed element from every solution to the reduced \Pigeon instance.
    
    We first consider the case where we obtain an element $x \in \{ 0, 1 \}^{n}$ such that $f(x) = v^*$.
    By definition, we know that $x \neq \pi$ since $f(\pi) = 1^{m_1} \| h(\pi) \neq 0^{m_1} \| h(\pi) = v^*$.
    It holds that $f(x) = p(x) \| h(x) = 0^{m_1} \| h(\pi)$, which implies that $h(x) = h(\pi)$. Thus, the element $x$ is a solution to the original instance.
    
    Next, we consider the case where we obtain two distinct elements $x$ and $y$ in $\{ 0, 1 \}^{n}$ such that $f(x) = f(y)$.
    By definition, the last $m_2$ bits of $f(x)$ and $f(y)$ are the same. Specifically, it holds that $h(x) = h(y)$.
    There are two cases when $x = \pi$ or not. If $x = \pi \neq y$, then we have that $h(y) = h(\pi)$, which implies that the element $y$ is a first type of solution to the original \CategorizedPigeon instance.
    Otherwise, i.e., if $x \neq \pi$, then $x$ is a first type of solution to the original instance if $y = \pi$; otherwise, $x$ and $y$ are a second type of solution to the original \CategorizedPigeon instance.
\end{proof}

Thus, it suffices to construct a polynomial-time reduction from \BasicWardSzabo to \CategorizedPigeon with exactly one removed element for proving \cref{Lemma:main_containment}.
Our reduction shown in \cref{Proof:Basic-to-Categorized} deals with the more special case of \CategorizedPigeon, in which $m_1 = m_2$.
For the sake of simplicity of our reduction, we introduce an alternative definition of \CategorizedPigeon satisfying that $m_1 = m_2$ and $k = 1$.
We modify the first type of solution to the \CategorizedPigeon into an element $x$ such that $h(x) = p(x)$.
Hence, in our new variant, the first-type solution does not depend on the removed element. Note that we consider the case where $m_1 = m_2$. Thus, we can define such a solution.
The formal definition of such a variant of \CategorizedPigeon, called \AlternativeCategorizedPigeon, is as follows:

\begin{definition} \label{Def:AlternativeCategorziedPigeon}
	The problem \AlternativeCategorizedPigeon is defined as follows: 
	Given two Boolean circuits $p: \{ 0, 1 \}^{2n} \to \{ 0, 1 \}^{n}$ and $h: \{ 0, 1 \}^{2n} \to \{ 0, 1 \}^{n}$ and a removed element $\pi \in \{ 0, 1 \}^{n}$, find one of the following: 
	\begin{enumerate*}[label = (\arabic*)]
		\item an element $x \in \{ 0, 1 \}^{n} \setminus \{ \pi \}$ such that $h(x) = p(x)$;
		\item two distinct elements $x, y \in \{ 0, 1 \}^{n} \setminus \{ \pi \}$ such that $p(x) = p(y)$ and $h(x) = h(y)$.
	\end{enumerate*}
\end{definition}

For proving \cref{Lemma:main_containment}, it is sufficient to prove that we have a polynomial-time reduction from \AlternativeCategorizedPigeon to \CategorizedPigeon with exactly one removed element.
The following lemma shows how to construct such a reduction. 
In fact, these two variants of \CategorizedPigeon are $\PPP$-complete from the $\PPP$-hardness of \WardSzabo shown in \cref{Sec:Hardness}.

\begin{lemma} \label{Lemma:Alternative-to-CategorizedOne}
	\AlternativeCategorizedPigeon is polynomial-time reducible to \CategorizedPigeon with exactly one removed element.
\end{lemma}
\begin{proof}
	Let $\mathcal{I} := \langle p: \{ 0, 1 \}^{2n} \to \{ 0, 1 \}^{n}; h: \{ 0, 1 \}^{2n} \to \{ 0, 1 \}^{n}; \pi \rangle$ be an \AlternativeCategorizedPigeon instance.
	We now construct a \CategorizedPigeon with exactly one removed element. In our reduction, the removed element for the reduced instance is the same as the original one.
	Also, the function $p$ is the same.
	Hence, we only modify the function $h$ in our reduction.
	
	We define the fucntion $\hat{h}: \{ 0, 1 \}^{2n} \to \{ 0, 1 \}^{n}$ as follows: For every element $x \in \{ 0, 1 \}^{2n}$,
	
	\begin{align*}
		\hat{h}(x) := \begin{cases}
			h(\pi) & \text{ if } x = \pi \text{ or } h(x) = p(x), \\
			p(x) & \text{ if } h(\pi) = h(x) \neq p(x), \\
			h(x) & \text{ otherwise.}
		\end{cases}
	\end{align*}
	
	We have completed constructing the reduced instance $\mathcal{J} := \langle p: \{ 0, 1 \}^{2n} \to \{ 0, 1 \}^{n}; \hat{h}: \{ 0, 1 \}^{2n} \to \{ 0, 1 \}^{n}; \pi \rangle$ of \CategorizedPigeon with exactly one removed element.
	It is not hard to see that our reduction can be computed in polynomial time.
	What remains is to prove that we can efficiently obtain a solution to $\mathcal{I}$ from every solution to $\mathcal{J}$.
	
	We first consdier the case where we obtain an element $x$ in $\{ 0, 1 \}^{2n} \setminus \{ \pi \}$ such that $\hat{h}(x) = \hat{h}(\pi)$.
	By definition, it satisfies that $\hat{h}(x) = \hat{h}(\pi) = h(\pi)$ and $x \neq \pi$. Hence, $h(x) = p(x)$ holds. Therefore, the element $x$ is a solution to the original instance $\mathcal{I}$.
	
	Next, we consider the case where we obtain two distinct elements $x$ and $y$ in $\{ 0, 1 \}^{2n} \setminus \{ \pi \}$ such that $p(x) = p(y)$ and $\hat{h}(x) = \hat{h}(y)$.
	If at least one of $h(x) = p(x)$ and $h(y) = p(y)$ holds, then we immediately solution to the original instance $\mathcal{I}$. 
	In the rest of this proof, we suppose that $h(x) \neq p(x)$ and $h(y) \neq p(y)$ hold.
	Then, we have two cases: If $h(x) = h(\pi)$ or not.
	When $h(x) = h(\pi)$, it holds that $\hat{h}(x) = p(x) = p(y)$. From our assumption that $h(x) \neq p(y)$, it must satisfy that $h(x) = h(\pi)$. 
	Therefore, the pair of two distinct elements $x$ and $y$ in $\{ 0, 1 \}^{2n}$ is a solution to the original \AlternativeCategorizedPigeon instance $\mathcal{I}$ since $h(x) = h(\pi) = h(y)$.
	On the other hand, when $h(x) \neq h(\pi)$, we have that $\hat{h}(x) = h(x) \neq p(x) = p(y)$. 
	Hence, it also holds that $h(y) \neq h(\pi)$.
	The pair of two distinct elements $x$ and $y$ is also a solution to $\mathcal{I}$ since $h(x) = \hat{h}(x) = \hat{h}(x) = h(y)$.
\end{proof}

Now, we are ready to prove that \BasicWardSzabo belongs to the complexity class $\PPP$.
The next lemma states that we have a polynomial-time reduction from \BasicWardSzabo to \AlternativeCategorizedPigeon.

\begin{lemma} \label{Lemma:Basic-to-Categorized}
	\BasicWardSzabo is reducible to \AlternativeCategorizedPigeon in polynomial time.
\end{lemma}

Here, we sketch the outline of our reduction.
The removed element is set to be the node $a$ for which the colors of two edges $\{ a, b \}$ and $\{ a, c \}$ are different. Note that the existence of such a node is guaranteed by the assumption of the problem \WardSzabo.
The partition function $p(x)$ is simply defined as the color on the edge $\{ a, x \}$ for every node $x$.
On the other hand, the hole function $h(x)$ is either the edge color on the edge $\{ b, x \}$ or $\{ c, x \}$.
The key idea of our reduction is to formulate the problem of finding four distinct nodes $x$, $y$, $z$, and $a$ satisfying the following conditions:

\begin{enumerate}[label = (\roman*)]
	\item the colors on the edges $\{ a, x \}$ and $\{ a, y \}$ are the same;
	\item the colors on the edges $\{ z, x \}$ and $\{ z, y \}$ are the same; 
	\item the colors on the edges $\{ a, z \}$ and $\{ z, x \}$ are the different; and
	\item the node $z$ is either $b$ or $c$.
\end{enumerate}

\noindent
We can immediately obtain a basic bichromatic triangle from such four distinct nodes as follows:
If the colors on the edges $\{ a, x \}$ and $\{ z, x \}$ are also the same, then the three nodes $a$, $x$, and $z$ form a bichromatic triangle; otherwise, at least one of the triangles $\{ a, x, y \}$ and $\{ z, x, y \}$ is bichromatic because the color on the edge $\{ x, y \}$ is different from at least one of the edges $\{ a, x \}$ and $\{ z, x \}$.
The full proof of this lemma can be found in \cref{Proof:Basic-to-Categorized}.

\cref{Lemma:main_containment} immediately follows from the above \cref{Lemma:WarSzabo-to-BasicWardszabo,Lemma:CategorizedOne-in-PPP,Lemma:Alternative-to-CategorizedOne,Lemma:Basic-to-Categorized}.

\begin{proof}[Proof of \cref{Lemma:main_containment}]
	From \cref{Lemma:WarSzabo-to-BasicWardszabo,Lemma:CategorizedOne-in-PPP,Lemma:Alternative-to-CategorizedOne,Lemma:Basic-to-Categorized}, we have a polynomial-time reduction from \WardSzabo to \Pigeon. Thus, the problem \WardSzabo belongs to $\PPP$.
\end{proof}

\subsubsection{Proof of \cref{Lemma:Basic-to-Categorized}}
\label{Proof:Basic-to-Categorized}
This section will prove \cref{Lemma:Basic-to-Categorized}.
Before proving this lemma, we provide a helpful observation in which we can assume that the edge-coloring function is symmetric without loss of generality.
In other words, we have an efficient reduction from an instance of \WardSzabo to another instance satisfying the symmetry.

\begin{proposition} \label{Prop:Assumption_symmetry}
	Without loss of generality, we can assume that every \WardSzabo (and also \BasicWardSzabo) instance $\mathcal{I} := \langle C: \{ 0, 1 \}^{2n} \times \{ 0, 1 \}^{2n} \to \{ 0, 1 \}^{n}; \{ a, b, c \} \rangle$ satisfies the symmetry, i.e., $C(x, y) = C(y, x)$ for each pair of two nodes $x$ and $y$ in $\{ 0, 1 \}^{2n}$.
\end{proposition}
\begin{proof}
	We modify the edge-coloring function $C$ as follows: For each pair of two nodes $x$ and $y$ in $\{ 0, 1 \}^{2n}$, we define $C'(x, y) := C(x, y)$ if $x$ is smaller than $y$ in the lexicographical order.
	We now obtain another \WardSzabo instance $\mathcal{J} := \langle C': \{ 0, 1 \}^{2n} \times \{ 0, 1 \}^{2n} \to \{ 0, 1 \}^{n}; \{ a, b, c \} \rangle$ in which $C'$ satisfies the symmetry.
	We can straightforwardly see that such an instance can be constructed in polynomial time.
	To complete the proof, we need to show that we can effortlessly obtain a solution to $\mathcal{I}$ from every solution to $\mathcal{J}$.
	
	Since the edge-coloring function $C'$ is symmetric, every solution to $\mathcal{J}$ is a bichromatic triangle $\{ x, y, z \}$.
	We first check whether each pair of two nodes $\xi$ and $\eta$ in $\{ x, y, z \}$ satisfies that $C(\xi, \eta) = C(\eta, \xi)$. If we find a pair of $\xi$ and $\eta$ such that $C(\xi, \eta) \neq C(\eta, \xi)$, then we obtain a violation that $C$ is symmetric, which is a solution to $\mathcal{I}$.
	Otherwise, the triangle $\{ x, y, z \}$ is also a bichromatic triangle. This is also a solution to $\mathcal{I}$.
\end{proof}

We show a polynomial-time reduction from \BasicWardSzabo to \AlternativeCategorizedPigeon.
Let $\mathcal{I} := \langle C: \{ 0, 1 \}^{2n} \times \{ 0, 1 \}^{2n} \to \{ 0, 1 \}^{n}; \{ a, b, c \} \rangle$ be an instance of \WardSzabo.
By \cref{Prop:Assumption_symmetry}, we can assume that every \BasicWardSzabo instance satisfies the symmetry without loss of generality.
By the definition of the problem \WardSzabo, it satisfies that $C(a, b) \neq C(b, c)$.

Now, we construct an \AlternativeCategorizedPigeon instance $\mathcal{J} := \langle p: \{ 0, 1 \}^{2n} \to \{ 0, 1 \}^{n}; \{ 0, 1 \}^{2n} \to \{ 0, 1 \}^{n}; \pi \rangle$.
First, we set a removed element $\pi$ to be $a$. Thus, the node $a$ is the removed element.
Next, we define the partition function $p: \{ 0, 1 \}^{2n} \to \{ 0, 1 \}^{n}$ as follows: For each element $x$, $p(x) = C(a, x)$, that is, $p(x)$ is the color on the edge $\{ a, x \}$.
Finally, we define the hole function $h: \{ 0, 1 \}^{2n} \to \{ 0, 1 \}^{n}$ as follows: For each element $x \in \{ 0, 1 \}^{2n}$,

\begin{align*}
    h(x) := \begin{cases}
        C(b, x) & \text{if } C(a, x) \neq C(a, b), \\
        C(c, x) & \text{otherwise.}
    \end{cases}
\end{align*}

\noindent
Thus, if the colors on the edges$\{ a, b \}$ and $\{ a, x \}$ are different, $h(x)$ is the color on the edge $\{ b, x \}$, otherwise it is the color on the edge $\{ c, x \}$.

We have completed constructing the \AlternativeCategorizedPigeon instance $\mathcal{J}$.
It is not hard to see that our reduction is computable in polynomial time.
What remains is to show that we can effortlessly obtain a solution to the original instance $\mathcal{I}$ from every solution to the reduced instance $\mathcal{J}$.

We first consider the case where we obtain an element $x \in \{ 0, 1 \}^{2n} \setminus \{ a \}$ such that $h(x) = p(x)$.
There are two cases in which $C(a, x) \neq C(a, b)$ holds or not.
First, we suppose that $C(a, x) \neq C(a, b)$ holds.
In this case, it satisfies that $C(b, x) = h(x) = p(x) = C(a, x) \neq C(a, b)$. 
This implies that the three nodes $a$, $b$, and $x$ form a bichromatic triangle, which is a solution to the original \BasicWardSzabo instance.
Next, we suppose that $C(a, x) = C(a, b)$ holds.
In this case, it satisfies that $C(c, x) = h(x) = p(x) = C(a, x) = C(a, b)$. Recall that we assume that $C(a, b) \neq C(a, c)$. Hence, the three nodes $a$, $c$, and $x$ form a bichromatic triangle, which is also a solution to the original instance.

We consider the case where we obtain two distinct elements $x, y \in \{ 0, 1 \}^{2n} \setminus \{ a \}$ such that $p(x) = p(y)$ and $h(x) = h(y)$.
By definition, $p(x) = p(y)$ implies that $C(a, x) = C(a, y)$, i.e., the color on the edges $\{ a, x \}$ and $\{ a, y \}$ are the same.
There are two cases in which $C(a, x) \neq C(a, b)$ or not.
First, we suppose that $C(a, x) \neq C(a, b)$. In this case, we have that $C(b, x) = h(x) = h(y) = C(b, y)$ since $C(a, y) = C(a, x) \neq C(a, b)$.
Furthermore, it also holds that $x \neq b \neq y$.
If $C(a, x) = C(b, x)$, then the three nodes $a$, $b$, and $x$ form a bichromatic triangle since $C(a, b) \neq C(a, x) = C(b, x)$.
Otherwise, i.e., if $C(a, x) \neq C(b, x)$, at least one of two triangles $\{ a, x, y \}$ and $\{ b, x, y \}$ is bichromatic because it holds that at least one of $C(a, x) = C(a, y) \neq C(x, y)$ and $C(b, x) = C(b, y) \neq C(x, y)$.
Therefore, we obtain a solution to the original instance $\mathcal{I}$.
Next, we suppose that $C(a, x) = C(a, b)$. In this case, it satisfies that $x \neq c \neq y$ since $C(a, y) = C(a, x) = C(a, b) \neq C(a, c)$, where the final inequality follows from the assumption of the \BasicWardSzabo instance.
Since $h(x) = h(y)$ and $C(a, y) = C(a, x) = C(a, b)$, we have that $C(c, x) = C(c, y)$.
If $C(c, x) = C(a, x)$, the triangle $\{ a, c, x \}$ is bichromatic since $C(a, x) \neq C(a, c)$. Otherwise, it satisfies at least one of $C(x, y) \neq C(a, x) = C(a, y)$ and $C(x, y) \neq C(c, x) = C(c, y)$. Hence, at least one of the two triangles $\{ a, x, y \}$ and $\{ c, x, y \}$ is bichromatic.

\subsection{Ward-Szabo is PPP-hard}
\label{Sec:Hardness}
In this section, we prove \cref{Lemma:main_hardness}.
Our proof originally comes from \cite{Ghentiyala25}. One of the heuristic observations made in \cite{BFHRS23}, and later formalized in \cite{Li24}, is that $\PPP$ problems which correspond to bounds can only be total when the corresponding bounds are sufficiently tight \cite{Ghentiyala25}.
So that, we utilize the tightness of the Ward-Szab\'{o} theorem (see \cref{Theorem:Ward-Szabo}).
Furthermore, we also use the $\PPP$-completeness of $\Pigeon_{2^{n}}^{2^n+1}$ (see \cref{Prop:pigeon_equiv}) for proving the $\PPP$-hardness of \WardSzabo.

\citeauthor{WS94} \cite{WS94} have explicitly shown a swell-coloring on edges of the complete graph with $p^{2k}$ nodes with $p^k + 1$ colors, where $p$ is prime (cf. Theorem 3 in \cite{WS94}).
We first construct a swell coloring on the complete graph with $2^{2n}$ nodes with $2^{n}$ colors based on their construction.
We need to observe at least the following two properties:
\begin{enumerate*}[label = (\roman*)]
    \item The edge-coloring is computable in polynomial time; and
    \item we have three distinct nodes $a, b, c$ such that the edge colors on $\{ a, b \}$ and $\{ a, c \}$ are different.
\end{enumerate*}

\begin{lemma}
    \label{Lemma:swell-circuit}
    There exists an efficiently constructible circuit $T: \{ 0, 1\}^{2n} \times \{ 0, 1\}^{2n} \rightarrow [2^n+1]$ and efficiently computable, distinct $a, b, c \in \{ 0, 1\}^{2n}$ such that
    \begin{enumerate}[label = (\arabic*)]
        \item $T(a, b) \neq T(a, c)$;
        \item $T(x, y) = T(y, x)$ for all $x, y \in \{ 0, 1\}^{2n}$; and
        \item there do not exist distinct $u, v, w \in \{ 0, 1\}^{2n}$ such that $T(u, v) = T(v, w) \neq T(u, w)$.
    \end{enumerate}
\end{lemma}
\begin{proof}
    Let $\mathbb{F}_{2^n}$ denote the Galois field of order $2^n$.
    For each element $u$ in $\{0, 1 \}^{2n}$, we denote $u_1$ and $u_2$ by the first and last $n$ bits of $u$, respectively, i.e., $u = u_1 \| u_2$.
    We regard each $n$-bit string as a field element in $\mathbb{F}_{2^n}$ using a canonical map from $\{ 0, 1 \}^{n}$ to $\mathbb{F}_{2^n}$.

    We now define the circuit $T: \{ 0, 1 \}^{2n} \times \{ 0, 1 \}^{2n} \to [2^n + 1]$ as follows:
    For each pair of two nodes $x$ and $y$ in $\{ 0, 1 \}^{2n}$,

    \begin{align*}
        T(x, y) := \begin{cases}
            2^n + 1 & \text{ if } x_1 = y_1, \\
            (x_2 - y_2) (x_1 - y_1)^{-1} & \text{ otherwise.}
        \end{cases}
    \end{align*}

    \noindent
    The circuit $T$ can be computed in polynomial time since we have efficient Turing machines that do these operations over the Galois fields.
    By using elements $0$ and $1$ in $\mathbb{F}_{2^n}$, we set the three distinct nodes $a$, $b$, and $c$ in $\{ 0, 1 \}^{2n}$ to be $1 \| 0$, $0 \| 1$, and $1 \| 1$, respectively.

    To complete the proof, let us move on to show that the circuit $T$ and three nodes $a, b, c$ satisfy the desired conditions.

    First, we check that $T(a, b) \neq T(a, c)$.
    By definition $T(a, b) = 2^n - 1 \in \mathbb{F}_{2^n}$, but $T(a, c) = 2^n + 1 \notin \mathbb{F}_{2^n}$, and thus, $T(a, b) \neq T(a, c)$.

    It is not hard to see that $T$ is symmetric, i.e., $T(x, y) = T(y, x)$ for each pair of two nodes $x$ and $y$ in $\{ 0, 1 \}^{2n}$.

    Finally, we show that there is no bichromatic triangle.
    To show this, we prove that $T(x, y) = T(x, z)$ implies that $T(y, z)$ for all three distinct nodes $x$, $y$, and $z$ in $\{ 0, 1 \}^{2n}$.
    If $T(x, y) = T(x, z) = 2^n+1$, then $y_1 = x_1 = z_1$. Therefore, it must hold that $T(y, z) = 2^n+1$.
    Otherwise, it satisfies that $T(x, y) = (x_2 - y_2)(x_1 - y_1)^{-1} = (x_2 - z_2)(x_1 - z_1)^{-1} = T(x, z)$ from the definition.
    Let $M = (x_2 - y_2)(x_1 - y_1)^{-1}$. We have the following two equations
    \begin{enumerate*}[label = (\roman*)]
        \item $x_2 - y_2 = M(x_1 - y_1)$, and
        \item $x_2 - z_2 = M(x_1 - z_1)$.
    \end{enumerate*}
    Then, we consider the difference between these two equation. 
    It holds that $y_2 - z_2 = M(y_1 - z_1)$. 
    This implies that $T(y, z) = (y_2 - z_2)(y_1 - z_1)^{-1} = M = (x_2 - y_2)(x_1 - y_1)^{-1} = T(x, y)$.
\end{proof}

We are now ready to prove \cref{Lemma:main_hardness}.
We will construct a polynomial-time reduction from $\Pigeon_{2^n}^{2^{n}+1}$ to \WardSzabo.
The key of our reduction is to map a swell-coloring constructed in \cref{Lemma:swell-circuit} to an edge-coloring with $2^n$ colors by using the function $f: [2^n+1] \to [2^n]$ given as an input instance of $\Pigeon_{2^n}^{2^{n}+1}$.

\begin{proof}[Proof of \cref{Lemma:main_hardness}]
    We show a polynomial-time reduction from $\Pigeon_{2^{n}}^{2^{n}+1}$ to \WardSzabo. From \cref{Prop:pigeon_equiv}, $\Pigeon_{2^{n}}^{2^{n}+1}$ is $\PPP$-complete.

    Let $f:[2^{n}+1] \to [2^n]$ be a $\Pigeon_{2^{n}}^{2^{n}+1}$ instance.
    We now construct an edge-coloring $C: \{ 0, 1 \}^{2n} \times \{ 0, 1 \}^{2n} \to \{ 0, 1 \}^{n}$ on the complete graph with $2^{2n}$ nodes.
    Let $T: \{ 0, 1 \}^{2n} \times \{ 0, 1 \}^{2n} \to [2^n + 1]$ be an efficiently computable swell-coloring on $K_{2^{2n}}$ shown in \cref{Lemma:swell-circuit}.
    Then, we define $C(x, y) := f(T(x, y))$ for every pair of two nodes $x$ and $y$ in $\{ 0, 1 \}^{2n}$.
    The reduced \WardSzabo instance is $\mathcal{J} := \langle C: \{ 0, 1 \}^{2n} \times \{ 0, 1 \}^{2n} \to \{ 0, 1 \}^{n}; a, b, c \rangle$, where the three distinct nodes $a$, $b$, and $c$ are specified in \cref{Lemma:swell-circuit}.

    We have completed constructing the reduction from $\Pigeon_{2^{n}}^{2^{n}+1}$ to \WardSzabo.
    It is not hard to see that our reduction can be computed in polynomial time.
    What remains is to prove that we obtain a solution to the original $\Pigeon_{2^{n}}^{2^{n}+1}$ instance from every solution to the reduced \WardSzabo instance $\mathcal{J}$.

    By definition, the coloring function $C$ is symmetric.
    Therefore, every solution to $\mathcal{J}$ is three distinct nodes $x$, $y$, and $z$ such that $C(x, y) = C(x, z) \neq C(y, z)$.
    From \cref{Lemma:swell-circuit}, the colors of $T(x, y)$, $T(x, z)$, and $T(x, z)$ are either the same or different from each other.
    If these three colors are the same, then it must hold that $C(x, y) = C(x, z) = C(y, z)$. Hence, these three colors are different from each other.
    On the other hand, we have that $C(x, y) = f(T(x, y)) = f(T(x, z)) = C(x, z)$. Thus, the pair of two distinct elements $T(x, y)$ and $T(x, z)$ in $[2^n + 1]$ is a solution to $\Pigeon_{2^{n}}^{2^{n}+1}$.
\end{proof}

From the $\PPP$-hardness of \WardSzabo, the $\PPP$-completeness of the two $\PPP$ problems \CategorizedPigeon with exactly one removed element and \AlternativeCategorizedPigeon immediately follows.

\begin{corollary}
    The two problems \AlternativeCategorizedPigeon and \CategorizedPigeon are $\PPP$-complete.
\end{corollary}

\subsection{A Parameterized Generalization of \WardSzabo}
\citeauthor{Ghentiyala25} \cite{Ghentiyala25} has introduced a parameterized generalization of \WardSzabo.
For $2n \le t(n) \le \poly(n)$, the problem $t(n)$-\WardSzabo is defined as follows: Given a Boolean circuit $C: \{ 0,1 \}^{t(n)} \times \{ 0, 1 \}^{t(n)}$ and three distinct nodes $a, b, c \in \{ 0, 1 \}^{t(n)}$ such that $C(a, b) \neq C(a, c)$, find one of the following:
\begin{enumerate*}[label = (\arabic*)]
    \item three distinct nodes $x$, $y$, and $z$ in $\{ 0, 1 \}^{t(n)}$ such that $C(x, y) = C(x, z) \neq C(y, z)$; and
    \item two nodes $x$ and $y$ in $\{ 0, 1 \}^{t(n)}$ such that $C(x, y) \neq C(y, x)$.
\end{enumerate*}

From our observation shown in \cref{Sec:UpperBound,Sec:Hardness}, we can easily see that $t(n)$-\WardSzabo is $\PWPP$-complete for any $2n < t(n) \le \poly(n)$.
For proving the $\PWPP$ containment, we formulate a weaker variant of \AlternativeCategorizedPigeon, called \WeakAlternativeCategorizedPigeon.

\begin{definition}
    The problem \WeakAlternativeCategorizedPigeon is defined as follows:
    Given two Boolean circuits $p: \{ 0, 1 \}^{n} \to \{ 0, 1 \}^{m}$ and $h: \{ 0, 1 \}^{n} \to \{ 0, 1 \}^{m}$ and a removed element $\pi \in \{ 0, 1 \}^{n}$, where $n > 2m$, find one of the following:
    \begin{enumerate*}[label = (\arabic*)]
        \item an element $x \in \{ 0, 1 \}^{n} \setminus \{ \pi \}$ such that $h(x) = p(x)$; and
        \item two distinct elements $x, y \in \{ 0, 1 \}^{n}$ such that $p(x) = p(y)$ and $h(x) = h(y)$.
    \end{enumerate*}
\end{definition}

\begin{theorem}
    The problem \WeakAlternativeCategorizedPigeon belongs to $\PWPP$.
\end{theorem}
\begin{proof}
    Let $\langle p: \{ 0, 1 \}^{n} \to \{ 0, 1 \}^{m}; h: \{ 0, 1 \}^{n} \to \{ 0, 1 \}^{m}; \pi \rangle$ be an instance of \WeakAlternativeCategorizedPigeon.
    We now construct a \WeakPigeon instance $\langle f: \{ 0, 1 \}^{n} \to \{ 0, 1 \}^{n-1} \rangle$.

    Let $t := n - 2m$.
    We define the function $f$ as follows: For every element $x \in \{ 0, 1 \}^{n}$,
    
    \begin{align*}
        f(x) := \begin{cases}
            0^{n-1} & \text{ if }  x = \pi, \\
            0^{t-1} \| p(x) \| h(x) & \text{ if } x \neq \pi.
        \end{cases}
    \end{align*}

    \noindent
    We have completed constructing the reduced \WeakPigeon instance.
    It is not hard to see that our reduction is computable in polynomial time.
    What remains is to show that we can efficiently obtain a solution to the original \WeakAlternativeCategorizedPigeon instance from every solution to the reduced \WeakPigeon instance.

    Consider that we obtain two distinct elements $x$ and $y$ in $\{ 0, 1 \}^{n}$ such that $f(x) = f(y)$.
    If $x = \pi$, then it holds that $p(y) = 0^m = h(y)$ and $y \neq \pi$. Therefore, the element $y$ is a solution to the original \WeakAlternativeCategorizedPigeon instance.
    On the other hand, if $x \neq \pi$, the $x$ is also a solution to the original instance if $y = \pi$; otherwise, it satifies that $f(x) = 0^{t-1} \| p(x) \| h(x) = 0^{t - 1} \| p(y) \| h(y) = f(y)$, which implies that $p(x) = p(y)$ and $h(x) = h(y)$ hold. Hence, the pair of two distinct elements $x$ and $y$ is a solution to the original instance.
\end{proof}

We can apply the same technique shown in \cref{Proof:Basic-to-Categorized} for showing a polynomial-time reduction from $t(n)$-\WardSzabo to \WeakAlternativeCategorizedPigeon for any $2n < t(n) \le \poly(n)$.
Also, we can apply the same approach shown in \cref{Sec:Hardness} to prove the $\PWPP$-hardness of $t(n)$-\WardSzabo.
As a consequence, we obtain the $\PWPP$-completeness of a parameterized \WardSzabo.

\begin{corollary}
    For $2n < t(n) \le \poly(n)$, the problem $t(n)$-\WardSzabo is $\PWPP$-complete.
\end{corollary}

\subsection{Basic Properties of Categorized Pigeonhole Principle}
\label{Sec:BasicProperties_of_CPPP}
The aim of this section is to observe the basic properties of a new $\TFNP$ problem \CategorizedPigeon.
We first show that the problem \CategorizedPigeon is a $\TFNP$ problem.

\begin{theorem}
	\CategorizedPigeon belongs to the class $\TFNP$.
\end{theorem}
\begin{proof}
	Let $\mathcal{I} := \langle p: \{ 0, 1 \}^{m_1 + m_2} \to \{ 0, 1 \}^{m_1}; h: \{ 0, 1 \}^{m_1 + m_2 } \to \{ 0, 1 \}^{m_2}; \Pi := \{ \pi_1, \dots, \pi_k \} \rangle$ be an instance of \CategorizedPigeon.
	We now show that there exists a solution to $\mathcal{I}$.
	
	For each $i$ in $\{ 0, 1 \}^{m_1}$, we denote by $M_{i}$ the set of all elements $x$ in $\{ 0, 1 \}^{m_1 + m_2}$ mapped $i$ by the function $p$. Thus, $M_{i} := \{ x \in \{ 0, 1 \}^{m_1 + m_2} : p(x) = i \}$.
	Note that $k < 2^{m_1}$.
	From the pigeonhole principle, we know that there is a set $M_{i^*}$ that has at least $2^{m_2}$ elements, i.e., $|M_{i^*}| \ge 2^{m_2}$.
	Then, we restrict the domain of the function $h$ to the set $M_{i^*}$. In other words, we consider the function $h_{i^*}: M_{i^*} \to \{ 0, 1 \}^{m_2}$ such that $h_{i^*}(x) = h(x)$ for every $x \in M_{i^*}$.
	Again, by the pigeonhole principle, we always have two distinct elements $x$ and $y$ in $M_{i^*}$ such that $h_{i^*}(x) = h_{i^*}(y)$ if $|M_{i^*}| > 2^{m_2}$. Such a pair of elements is a solution to $\mathcal{I}$ since $h_{i^*}(x) = h(x) = h(y) = h_{i^*}(y)$ and $M_{i^*} \subseteq \{ 0, 1 \}^{m_1 + m_2}$.
	
	Next, we consider the case where $|M_{i^*}| = 2^{m_2}$.
	If there is $a \in \{ 0, 1 \}^{m_2}$ such that $h_{i^*}(x) \neq a$ for every $x \in \{ 0, 1 \}^{m_2}$, then we have two distinct elements $x$ and $y$ in $M_{i^*}$ satisfying that $h_{i^*}(x) = h_{i^*}(y)$ from the pigeonhole principle.
	We also have a solution to the \CategorizedPigeon instance $\mathcal{I}$.
	Otherwise, i.e., if $h_{i^*}$ is bijective, there always exists an element $x \in \{ 0, 1 \}^{m_1 + m_2}$ such that the function $h_{i^*}$ maps $x$ to $h(\pi) \in \{ 0, 1 \}^{m_2}$. We have a solution to $\mathcal{I}$ since $h_{i^*}(x) = h(x) = h(\pi)$.
	
	From the above argument, every \CategorizedPigeon instance always has a solution.
	Moreover, for every candidate solution, its correctness is verifiable in polynomial time. 
	Therefore, the problem \CategorizedPigeon belongs to the complexity class $\TFNP$.
\end{proof}

We know the $\PPP$-completeness of \CategorizedPigeon with $k = 1$.
Let us move on to observe the computational complexity of \CategorizedPigeon with at least two removed elements.
In the remainder of this section, we assume that $k > 1$.
We first show that there is a polynomial-time reduction from \Pigeon to \CategorizedPigeon for any $k \in \poly(n)$, i.e., $k$ is at most polynomial in the number of inputs of the Boolean circuit $f$ given as input instance for \Pigeon.

\begin{theorem}
	\CategorizedPigeon is $\PPP$-hard.
\end{theorem}
\begin{proof}
	We show a polynomial-time reduction from \Pigeon to \CategorizedPigeon.
	Let $\langle f: \{ 0, 1 \}^{n} \to \{ 0, 1 \}^{n}; v^* \rangle$ be an instance of \Pigeon.
	
	Let $k \in \poly(n)$ be any positive integer.
	Let $m_1 = k$ and $m_2 = n$.
	We define the set of removed elements $\Pi$ as follows: $\Pi := \{ (1^{i}0^{k-i}, 0^{n}) : \forall i \in [k] \}$.
	Finally, we define the two Boolean circuits $p: \{ 0, 1 \}^k \times \{ 0, 1 \}^{n} \to \{ 0, 1 \}^{k}$ and $h: \{ 0, 1 \}^{k} \times \{ 0, 1 \}^{n} \to \{ 0, 1 \}^{n}$ as follows: For every $(b, x) \in \{ 0, 1 \}^k \times \{ 0, 1 \}^{n}$, $p(b, x) := b$ and 
	\begin{align*}
		h(b, x) := \begin{cases}
			v^{*} & \text{ if } (b, x) \in \Pi, \\
			f(x) & \text{ otherwise.}
		\end{cases}
	\end{align*}
	Thus, the function $p$ simply outputs the first $k$ bits of the input.
	On the other hand, the function $h$ outputs the special element $v^*$ if the input is a removed element; otherwise, $h$ simulates the function $f$ using the last $n$ bits of the input.

	We have completed constructing the \CategorizedPigeon instance.
	It is not hard to see that our reduction is polynomial-time computable.
	What remains is to prove that we can effortlessly obtain a solution to the original \Pigeon instance from every solution to the reduced \CategorizedPigeon instance.
	
	First, we consider the case where we obtain an element $(b, x) \in \{ 0, 1 \}^{k} \times \{ 0, 1 \}^{n} \setminus \Pi$ such that $h(b, x) \in h(\Pi)$.
	By definition, it holds that $f(x) = h(b, x) = h(1^k, 0^n) = v^*$.
	Thus, the element $x$ in $\{ 0, 1 \}^{n}$ is a solution to the original \Pigeon instance.
	
	Next, we consider the case where we obtain two distinct elements $(b, x), (c, y) \in \{ 0, 1 \}^k \times \{ 0, 1 \}^{n} \setminus \Pi$ such that $p(b, x) = p(c, y)$ and $h(b, x) = h(c, y)$.
	In this case, it holds that $b = c$, and $x \neq y$.
	By definition, it satisfies that $f(x) = h(b, x) = h(c, y) = f(x)$.
	Therefore, the pair of two distinct elements $x$ and $y$ in $\{ 0, 1 \}^{n}$ is a solution to the original \Pigeon instance.
\end{proof}

\noindent
It is unclear whether we have a polynomial-time reduction from \CategorizedPigeon to \Pigeon when $k > 1$ because we are unaware of how to avoid bad solutions in which a pair of two elements in $\Pi$ is a collision type of solution to a reduced \Pigeon instance.
In particular, when we have two distinct removed elements $\pi$ and $\pi'$ in $\Pi$ such that $h(\pi) = h(\pi')$, it becomes unclear whether \CategorizedPigeon belongs to the class $\PPP$.

There is an idea that we introduce a restricted variant of \CategorizedPigeon in which for each pair of removed elements $\pi$ and $\pi'$, it holds that $h(\pi) \neq h(\pi')$.
We call such a variant \InjectiveCategorizedPigeon (see \cref{Def:InjectiveCategorziedPigeon}).
\cref{Theorem:Injective-in-PPP} states that the problem \InjectiveCategorizedPigeon is in the complexity class $\PPP$.

\begin{definition} \label{Def:InjectiveCategorziedPigeon}
	The problem \InjectiveCategorizedPigeon is defined as follows: 
	Given two Boolean circuits $p: \{ 0, 1 \}^{m_1 + m_2} \to \{ 0, 1 \}^{m_1}$ and $h: \{ 0, 1 \}^{m_1 + m_2} \to \{ 0, 1 \}^{m_2}$ and a list of $k$ removed elements $\Pi := \{ \pi_1, \pi_2, \dots, \pi_k \} \subseteq \{ 0, 1 \}^{m_1 + m_2}$ such that $h(\pi_i) \neq h(\pi_j)$ for each pair of distinct two removed elements $\pi_i, \pi_j \in \Pi$, where $k < 2^{m_1}$, find one of the following: 
	\begin{enumerate*}[label = (\arabic*)]
		\item an element $x \in \{ 0, 1 \}^{n} \setminus \{ \pi \}$ such that $h(x) \in h(\Pi)$;
		\item two distinct elements $x, y \in \{ 0, 1 \}^{n} \setminus \Pi$ such that $p(x) = p(y)$ and $h(x) = h(y)$.
	\end{enumerate*}
\end{definition}

\begin{theorem} \label{Theorem:Injective-in-PPP}
	\InjectiveCategorizedPigeon belongs to $\PPP$.
\end{theorem}
\begin{proof}
	Let $\mathcal{I} := \langle p: \{ 0, 1 \}^{m_1 + m_2} \to \{ 0, 1 \}^{m_1}; h: \{ 0, 1 \}^{m_1 + m_2} \to \{ 0, 1 \}^{m_2}; \Pi:= \{ \pi_1, \dots, \pi_k \} \rangle$ be an instance of \InjectiveCategorizedPigeon.
	Let $n := m_1 + m_2$
	We now construct a \Pigeon instance $\mathcal{J} := \langle f: \{ 0, 1 \}^n \to \{ 0, 1 \}^n; v^* \rangle$ from the instance $\mathcal{I}$.
	
	We set the special element $v^*$ to be $0^{m_1} \| h(\pi_1)$ in $\{ 0, 1 \}^{n}$ if $p(\pi_1) \neq 0^{m_1}$; otherwise, to be $1^{m_1} \| h(\pi_1)$.
	We define the Boolean circuit $f: \{ 0, 1 \}^{n} \to \{ 0, 1 \}^{n}$ as follows: For every $x \in \{ 0, 1 \}^{n}$, $f(x) := p(x) \| h(x)$. Thus, for each input string $x$ in $\{ 0, 1 \}^n$, the first $m_1$-bits and the last $m_2$-bits of $f(x)$ are $p(x)$ and $h(x)$, respectively.
	We have completed constructing the reduced instance $\mathcal{J}$.
	It is not hard to see that our reduction is polynomial-time computable.
	What remains is to show that we efficiently obtain a solution to the original \InjectiveCategorizedPigeon instance $\mathcal{I}$ from every solution to the reduced \Pigeon instance $\mathcal{J}$.
	
	We first consider the case where we obtain an element $x$ in $\{ 0, 1 \}^n$ such that $f(x) = v^*$.
	Since for every $i > 1$, $h(\pi_i) \neq h(\pi_1)$, and since $f(\pi_1) \neq v^*$, the element $x$ is not in $\Pi$.
	By definition, the last $m_2$-bits of $f(x)$ and $v^*$ are $h(x)$ and $h(\pi_1)$, respectively. Hence, the element $x$ is a solution to the original \InjectiveCategorizedPigeon instance $\mathcal{I}$ since $h(x) = h(\pi_1)$.
	
	Next, we consider the case where we obtain two distinct elements $x$ and $y$ in $\{ 0, 1 \}^{n}$ such that $f(x) = f(y)$.
	By defintion, it holds that $p(x) = p(y)$ and $h(x) = h(y)$.
	Since for every pair of two disctinct removed elements $\pi_i$ and $\pi_j$ in $\Pi$,  $h(\pi_i) \neq h(\pi_j)$, at least one of $x$ and $y$ is not in $\Pi$. Without loss of generality, we can assume that $x \notin \Pi$.
	If $y \in \Pi$, the element $x$ is a solution to $\mathcal{I}$ since $h(x) \in h(\Pi)$.
	Otherwise, the pair of two distinct elements $x$ and $y$ is a solution to $\mathcal{I}$.
\end{proof}

Unfortunately, this time, it is unclear whether we will be able to prove the $\PPP$-hardness of \InjectiveCategorizedPigeon.
Recall the $\PPP$-hardness proof of \CategorizedPigeon, the function $h$ is defined so that every element in $\Pi$ is mapped to the same element. This formulation violates the assumption for an instance of \InjectiveCategorizedPigeon.
Hence, we need a more clever idea to show the $\PPP$-completeness of \CategorizedPigeon and \InjectiveCategorizedPigeon.

To conclude this paper, we provide a complexity lower bound for \InjectiveCategorizedPigeon.
Specifically, we prove the $\PPADS$-hardness of this problem; the complexity class $\PPADS$ is one of the subclasses of $\PPP$ \cite{Pap94}.
Originally, the complexity class $\PPADS$ \cite{Pap94} is defined as the set of all search problems that are reducible to \SinkofLine in polynomial time.

\begin{definition}
	The problem \SinkofLine is defined as follows: Given two Boolean circuits $S, P: \{ 0, 1 \}^n \to \{ 0, 1 \}^n$ and a known source $s \in \{ 0, 1 \}^{n}$ such that $P(s) = s \neq S(s)$, find a sink $x$ in $\{ 0, 1 \}^{n}$ such that $P(S(x)) \neq x$.
\end{definition}

For every \SinkofLine instance, we consider the digraph with $2^n$ nodes whose nodes have at most one in-degree and out-degree.
For every node, the functions $S$ and $P$ represent a successor and a predecessor, respectively.
Thus, there is a directed edge from $u$ to $v$ if $S(u) = v$ and $P(v) = u$.
Since we know the existence of a source node, there must exist a sink node. The computational task of \SinkofLine is to seek a sink node. Note that a solution does not necessary have to be reachable from the known source node.

For proving the $\PPADS$-hardness of \InjectiveCategorizedPigeon, we use another $\PPADS$-complete problem, a multi-source variant of \SinkofLine \cite{HG18}.

\begin{definition}
	The problem \MultiSourceSinkofLine is defined as follows: Given two Boolean circuits $S, P: \{ 0, 1 \}^{n} \to \{ 0, 1 \}^{n}$ and a list of known sources $\Sigma:= \{ s_1, s_2, \dots, s_k \} \subseteq \{ 0, 1 \}^{n}$ such that $P(s) = s \neq S(s)$ for each $s$ in $\Sigma$, where $1 \le k \le \poly(n)$, find a sink $x$ in $\{ 0, 1 \}^{n}$ such that $P(S(x)) \neq x$.
\end{definition}

\begin{theorem}[\citeauthor{HG18} \cite{HG18}]
	The problem \MultiSourceSinkofLine is $\PPADS$-complete.
\end{theorem}

\begin{theorem}
	\InjectiveCategorizedPigeon is $\PPADS$-hard.
\end{theorem}
\begin{proof}
	We will show a polynomial-time reduction from the $\PPADS$-hard problem \MultiSourceSinkofLine to \InjectiveCategorizedPigeon.
	Let $\mathcal{I} := \langle S, P: \{ 0, 1 \}^{n} \to \{ 0, 1 \}^{n}; \Sigma := \{ s_1, s_2, \dots, s_k \} \rangle$ be an instance of \MultiSourceSinkofLine.
	For every node $x$ in $\{ 0, 1 \}^n$, we say that $x$ is \textit{valid} if $S(x) = x$ or $S(x) \neq x$ and $P(S(x)) = x$, and we say that $x$ is \textit{invalid} if $x$ is not valid. Note that every invalid node is a solution to \SinkofLine instance.
	Without loss of generality, we can assume that every known source is valid since we are interested in the non-trivial problem. If we have a known source $s$ that is invalid, then we immediately obtain a solution to the \MultiSourceSinkofLine instance since $P(S(s)) \neq s$.
	
	Let $m_1 = k$ and $m_2 = n$.
	We now construct an instance of \InjectiveCategorizedPigeon.
	We first define the set of removed elements $\Pi$ as follows: $\Pi := \{ (1^{i}0^{k-i}, s_i) : \forall i \in [k], s_i \in \Sigma \}$.
	We define the functions $p: \{ 0, 1 \}^k \times \{ 0, 1 \}^{n} \to \{ 0, 1 \}^k$ and $h: \{ 0, 1 \}^k \times \{ 0, 1 \}^{n} \to \{ 0, 1 \}^{n}$ as follows: For every $(b, x) \in \{ 0, 1 \}^k \times \{ 0, 1 \}^n$, $p(b, x) := b$ and
	\begin{align*}
		h(b, x) :=  \begin{cases}
				x & \text{ if } (b, x) \in \Pi, \\
				S(x) & \text{ if } (b, x) \notin \Pi \text{ and $x$ is valid}, \\
				s_1 & \text{ otherwise.}  		
  		\end{cases}
	\end{align*}	
	The function $p$ always outputs the first $k$ bits of the input.
	On the other hand, the function $h$ basically simulates the successor $S$ when the input is not a removed element, and the last $n$ bits represent a valid node.
	Otherwise, the function $h$ returns some known source.
	
	We have completed constructing the reduced \InjectiveCategorizedPigeon instance.
	It is not hard to see that our reduction is polynomial-time computable.
	What remains is to show that we can effortlessly obtain a solution to the original \MultiSourceSinkofLine instance $\mathcal{I}$ from every solution to the reduced \InjectiveCategorizedPigeon instance.
	
	First, we consider the case where we obtain an element $(b, x)$ in $\{ 0, 1 \}^k \times \{ 0, 1 \}^{n} \setminus \Pi$ such that $h(b, x) \in h(\Pi)$.
	By definition, it holds that $h(b, x) \in \Sigma$.
	Assuming that $v$ is valid, it holds that $h(b, x) = S(x) \neq x$ since $s \neq S(s) \notin \Sigma$ for every known source $s$ in $\Sigma$.
	Moreover, it satisfies that $P(s) = s$ for each $s \in \Sigma$. This is a contradiction that $x$ is valid.
	Therefore, $x$ is invalid. Since $S(x) \neq x$ and $P(S(x)) \neq x$ hold, the node $x$ is a solution to the original instance $\mathcal{I}$.
	
	Next, we consider the case where we obtain a pair of two distinct elements $(b_1, x_1)$ and $(b_2, y_2)$ in $\{ 0, 1 \}^k \times \{ 0, 1 \}^n$ such that $p(b_1, x_1) = p(b_2, x_2)$ and $h(b_1, x_1) = h(b_2, x_2)$.
	By definition, it holds that $b_1 = b_2$ and $x_1 \neq x_2$.
	We immediately obtain a solution to $\mathcal{I}$ if at least one of $x$ and $y$ is invalid. Hence, in the rest of this proof, we suppose that both $x$ and $y$ are valid.
	Then, exactly one of the following holds:
	\begin{enumerate*}[label = (\roman*)]
		\item $S(x_1) = x_1$, $S(x_2) = x_1 \neq x_2$, and $P(S(x_2)) = x_2$; or
		\item $S(x_1) = x_2 \neq x_1$, $P(S(x_1)) = x_1$, and $S(x_2) = x_2$.
	\end{enumerate*}	
	Without loss of generality, we can assume that the former case holds.
	This implies that the node $x_1$ in $\{ 0, 1 \}^{n}$ is a solution to the original instance $\mathcal{I}$.
\end{proof}




\section*{Acknowledgment}
This work was supported by JSPS KAKENHI, Grant-in-Aid for Early-Career Scientists, Grant Number JP25K21155.

The author deeply appreciates the helpful comments of Surendra Ghentiyala, who generously shared the $\PPP$-hardness proof of \WardSzabo and offered thoughtful feedback. This feedback has greatly contributed to the improvement of this paper.

\printbibliography

@inproceedings{BGS25,
  author       = {Huck Bennett and
                  Surendra Ghentiyala and
                  Noah Stephens{-}Davidowitz},
  title        = {The More the Merrier! On Total Coding and Lattice Problems and the Complexity of Finding Multicollisions},
  booktitle    = {16th Innovations in Theoretical Computer Science Conference},
  series       = {LIPIcs},
  volume       = {325},
  pages        = {14:1--14:22},
  year         = {2025},
  doi          = {10.4230/LIPICS.ITCS.2025.14}
}

@inproceedings{BFHRS23,
  author       = {Romain Bourneuf and
                  Luk\'{a}v{s} Folwarczn\'{y} and
                  Pavel Hub{\'{a}}cek and
                  Alon Rosen and
                  Nikolaj I. Schwartzbach},
  title        = {PPP-Completeness and Extremal Combinatorics},
  booktitle    = {14th Innovations in Theoretical Computer Science Conference},
  series       = {LIPIcs},
  volume       = {251},
  pages        = {22:1--22:20},
  year         = {2023},
  doi          = {10.4230/LIPICS.ITCS.2023.22}
}

@inproceedings{FGPR24,
  author       = {Noah Fleming and
                  Stefan Grosser and
                  Toniann Pitassi and
                  Robert Robere},
  title        = {Black-Box {PPP} Is Not Turing-Closed},
  booktitle    = {Proceedings of the 56th Annual {ACM} Symposium on Theory of Computing},
  pages        = {1405--1414},
  publisher    = {{ACM}},
  year         = {2024},
  doi          = {10.1145/3618260.3649769}
}

@misc{Ghentiyala25,
    author = {Surendra Ghentiyala},
    title ={Private Communication},
    year = {2025}
}

@article{HG18,
  author       = {Alexandros Hollender and
                  Paul W. Goldberg},
  title        = {The Complexity of Multi-source Variants of the End-of-Line Problem, and the Concise Mutilated Chessboard},
  journal      = {Electronic Colloquium on Computational Complexity},
  volume       = {{TR18-120}},
  year         = {2018},
  url          = {https://eccc.weizmann.ac.il/report/2018/120}
}

@inproceedings{HV21,
  author       = {Pavel Hub{\'{a}}cek and
                  Jan V{\'{a}}clavek},
  title        = {On Search Complexity of Discrete Logarithm},
  booktitle    = {46th International Symposium on Mathematical Foundations of Computer Science},
  series       = {LIPIcs},
  volume       = {202},
  pages        = {60:1--60:16},
  year         = {2021},
  doi          = {10.4230/LIPICS.MFCS.2021.60}
}

@article{Ishizuka25,
    title = {On the complexity of some restricted variants of Quotient Pigeon and a weak variant of K{\H{o}}nig},
    journal = {Information Processing Letters},
    volume = {190},
    pages = {106574},
    year = {2025},
    doi = {https://doi.org/10.1016/j.ipl.2025.106574},
    author = {Takashi Ishizuka}
}

@article{Jer16,
    title = {Integer factoring and modular square roots},
    journal = {Journal of Computer and System Sciences},
    volume = {82},
    number = {2},
    pages = {380-394},
    year = {2016},
    issn = {0022-0000},
    doi = {https://doi.org/10.1016/j.jcss.2015.08.001},
    author = {Emil Jeřábek}
}

@book{Jukna11,
    title={Extremal combinatorics: with applications in computer science},
    author={Jukna, Stasys},
    volume={571},
    year={2011},
    publisher={Springer},
    doi = {https://doi.org/10.1007/978-3-642-17364-6}
}

@inproceedings{JLRX24,
  author       = {Siddhartha Jain and
                  Jiawei Li and
                  Robert Robere and
                  Zhiyang Xun},
  title        = {On Pigeonhole Principles and Ramsey in {TFNP}},
  booktitle    = {65th {IEEE} Annual Symposium on Foundations of Computer Science},
  pages        = {406--428},
  publisher    = {{IEEE}},
  year         = {2024},
  doi          = {10.1109/FOCS61266.2024.00033}
}

@inproceedings{Li24,
  author       = {Jiawei Li},
  title        = {Total {NP} Search Problems with Abundant Solutions},
  booktitle    = {15th Innovations in Theoretical Computer Science Conference},
  series       = {LIPIcs},
  volume       = {287},
  pages        = {75:1--75:23},
  year         = {2024},
  doi          = {10.4230/LIPICS.ITCS.2024.75}
}

@article{Pap94,
  author       = {Christos H. Papadimitriou},
  title        = {On the Complexity of the Parity Argument and Other Inefficient Proofs
                  of Existence},
  journal      = {Journal of Computer and System Sciences},
  volume       = {48},
  number       = {3},
  pages        = {498--532},
  year         = {1994},
  doi          = {10.1016/S0022-0000(05)80063-7}
}

@inproceedings{PPY23,
  author       = {Amol Pasarkar and
                  Christos H. Papadimitriou and
                  Mihalis Yannakakis},
  title        = {Extremal Combinatorics, Iterated Pigeonhole Arguments and Generalizations of {PPP}},
  booktitle    = {14th Innovations in Theoretical Computer Science Conference},
  series       = {LIPIcs},
  volume       = {251},
  pages        = {88:1--88:20},
  year         = {2023},
  doi          = {10.4230/LIPICS.ITCS.2023.88}
}

@inproceedings{SZZ18,
  author       = {Katerina Sotiraki and
                  Manolis Zampetakis and
                  Giorgos Zirdelis},
  title        = {{PPP}-Completeness with Connections to Cryptography},
  booktitle    = {59th {IEEE} Annual Symposium on Foundations of Computer Science},
  pages        = {148--158},
  publisher    = {{IEEE} Computer Society},
  year         = {2018},
  doi          = {10.1109/FOCS.2018.00023}
}

@article{WS94,
  title={On swell colored complete graphs},
  author={Ward, C and Szab{\'o}, S},
  journal={Acta Math. Univ. Comenianae},
  volume={63},
  number={2},
  pages={303--308},
  year={1994}
}

\end{document}